\def\rank{{\rm rank}\mathop}

\documentclass[10pt,conference]{IEEEtran}
\usepackage{cite}
\usepackage{mathrsfs}
\usepackage{hyperref}
\usepackage{amsthm}
\usepackage{amsmath}
\usepackage{amssymb}
\usepackage{revsymb}
\usepackage{graphicx}
\def\mat#1{\mathcal{#1}}

\newtheorem{theorem}{Theorem}

\begin{document}

\title{Linked-Cluster Technique for Finding the Distance of a Quantum
  LDPC Code}

\author{\authorblockN{Alexey A. Kovalev}
  \authorblockA{Department of Physics \& Astronomy\\
    University of California\\
    Riverside, CA 92521, USA \\
    Email: alexey.kovalev@ucr.edu}
\and 
\authorblockN{Ilya Dumer}\IEEEmembership{Fellow,˜IEEE}
  \authorblockA{Department of Electrical Engineering\\
    University of California\\
    Riverside, CA 92521, USA \\
    Email: dumer@ee.ucr.edu}
  \and
\authorblockN{Leonid P. Pryadko}
\authorblockA{Department of Physics \& Astronomy\\
University of California\\
Riverside, CA 92521, USA \\
Email: leonid@ucr.edu}
}

\maketitle

\begin{abstract}
  We present a linked-cluster technique for calculating the distance
  of a quantum LDPC code.  It offers an advantage over existing
  deterministic techniques for codes with small relative distances
  (which includes all known families of quantum LDPC codes), and over
  the probabilistic technique for codes with sufficiently high rates.
\end{abstract}

\section{Introduction}
A practical implementation of a quantum computer will rely on quantum error
correction (QEC) \cite{shor-error-correct,Knill-Laflamme-1997,Bennett-1996}
due to the fragility of quantum states.  There is a strong belief that surface
(toric) codes \cite{kitaev-anyons,Dennis-Kitaev-Landahl-Preskill-2002} can
offer the fastest route to scalable quantum computation due to the error
threshold around 1\% and the locality of required gates
\cite{Raussendorf-Harrington-2007,Wang-Fowler-Austin-Hollenberg-Lloyd-2011,Fowler:PRA2012,Bombin-PRX-2012}.
Unfortunately, in the nearest future, the surface codes (in fact, any
two-dimensional codes with local stabilizer
generators\cite{Bravyi-Poulin-Terhal-2010}) can only lead to proof of the
principle realizations as they encode a limited number of qubits ($k$), making
any implementation of a useable quantum computer large (e.g., $2.2\times10^8$
physical qubits are required for a useful realization of Shor's algorithm
\cite{Fowler:2012arXiv}).

Lifting the restriction of locality but preserving the condition that
the stabilizer generators should only involve a limited number of
qubits, one gets the quantum LDPC codes, or, more precisely, quantum
sparse-graph codes\cite{Postol-2001,MacKay-Mitchison-McFadden-2004}.
Unlike the surface codes, these more general quantum LDPC codes can
have a finite rate.  On the other hand, while there are no known upper
bounds on the parameters of such codes, in practice, all families of
quantum LDPC codes where the upper limit on the distance is known,
have the distance scaling as a square root of the block
length\cite{Tillich2009,%
  Kovalev-Pryadko-2012,Kovalev-Pryadko-Hyperbicycle-2012,
  Andriyanova-Maurice-Tillich-2012}.  Nevertheless, such codes (in
fact, any family of quantum or classical LDPC codes with limited
weights of the columns and rows of the parity check matrix, and
distance scaling as a power or a logarithm of the block length $n$) have a
finite error probability threshold, both in the standard setting where
syndrome is measured exactly, and with the syndrome measurement
errors\cite{Kovalev-Pryadko-FT-2012}.  

Given that non-local two-qubit gates are relatively inexpensive with
floating gates\cite{Loss:PRX2012}, superconducting and trapped-ion
qubits, as well as more exotic schemes with
teleportation\cite{yamamoto-cnot-2003,%
  Martinis-science-2005,Benhelm-2008,Friedenauer-2008,%
  Bennett-teleportation-1993,%
  Gottesman-Chuang-1999}, a quantum computer relying on quantum LDPC
codes is quite feasible.  An example of a universal set of gates based
on dynamical decoupling pulses for an arbitrary number of qubits with
Ising couplings forming a bipartite graph (e.g., the Tanner graph
corresponding to a quantum LDPC code) has been recently suggested by
one of us\cite{De-Pryadko-2012}.  

Compared to general quantum codes, with a quantum LDPC code, each
quantum measurement involves fewer qubits, measurements can be done in
parallel, and also the classical processing could potentially be
enormously simplified (note, however, that belief-propagation and
related decoding algorithms that work so well for classical LDPC codes\cite{Gallager1962,MacKay:2002} may
falter in the quantum case\cite{Poulin-Chung-2008}).  Compared to
surface codes, more general LDPC codes have higher rates, which
translates in a large reduction of the total number of qubits
necessary to build a useful quantum computer.  Note that while our
analytical threshold estimate in Ref.~\cite{Kovalev-Pryadko-FT-2012}
is quite low, there are examples of quantum LDPC codes demonstrated to
beat the bounded distance decoding
limit\cite{Kasai-Hagiwara-Imai-Sakaniwa-2012}.  Overall, it is quite
plausible that the operation of quantum computers of the future will
rely on (non-local) quantum LDPC codes.

The very general proof\cite{Kovalev-Pryadko-FT-2012} of the existence
of a finite error probability threshold for quantum and classical LDPC
codes with asymptotically zero relative distance is based on a simple
observation that errors for such codes are likely to form small
clusters affecting disjoint sets of stabilizer generators (parity
check matrix rows).  While the total weight of an error could be huge,
the error can be surely detected if the size of each cluster is
smaller than the code distance.  Thus, in the case of the error
detection, the threshold problem is related to the cluster size
distribution for site percolation on a graph related to the Tanner
graph of the code.

In this work, we apply the idea of error clustering with LDPC codes to
design a numerical algorithm for finding a distance of such a code.
The basic principle is formulated in Theorem \ref{th:basic}: to find
the distance of a code, one only needs to check error configurations
corresponding to connected error clusters.  For any error weight $w\ll
n$, the number of such clusters is exponentially smaller than that of
generic errors of the same weight.  We consider the complexity of
several well-known classical algorithms for finding code distance in
application to quantum error correcting code.  We conclude that the
clustering algorithm beats deterministic techniques at sufficiently
small relative distances (asymptotically at large $n$, all known
families of quantum LDPC codes have zero relative distance), and the
probabilistic technique for high-rate codes with small relative
distances.

\section{Background.}
\subsection{Error-correcting codes}
A $q$-ary linear code $\mathcal{C}$ with parameters $[n,k,d]_q$ is a
$k$-dimensional subspace of the vector space $\mathbb{F}_q^n$ of all
$q$-ary strings of length $n$.  Code distance $d$ is the minimal
Hamming weight (number of non-zero elements) of a non-zero string in
the code.  A linear code is uniquely specified by the parity check
matrix $H$, namely $\mathcal{C}=\{\mathbf{c}\in \mathbb{F}_q^n| H
\mathbf{c}=0\}$, where operations are done according to the
$\mathbb{F}_q$ algebra.

A quantum $[[n,k,d]]$ (qubit) stabilizer code $\mathcal{Q}$ is a
$2^k$-dimensional subspace of the $n$-qubit Hilbert space
$\mathbb{H}_{2}^{\otimes n}$, a common $+1$ eigenspace of all
operators in an Abelian \emph{stabilizer group}
$\mathscr{S}\subset\mathscr{P}_{n}$, $-\openone\not\in\mathscr{S}$,
where the $n$-qubit Pauli group $\mathscr{P}_{n}$ is generated by
tensor products of the $X$ and $Z$ single-qubit Pauli operators.  The
stabilizer is typically specified in terms of its generators,
$\mathscr{S}=\left\langle S_{1},\ldots,S_{n-k}\right\rangle $;
measuring the generators $S_i$ produces the \emph{syndrome} vector.
The weight of a Pauli operator is the number of qubits it affects.
The distance $d$ of a quantum code is the minimum weight of an
operator $U$ which commutes with all operators from the stabilizer
$\mathscr{S}$, but is not a part of the stabilizer,
$U\not\in\mathscr{S}$. A code of distance $d$ can detect any error of
weight up to $d-1$, and correct up to $\lfloor d/2\rfloor$.

A Pauli operator $U\equiv i^{m}X^{\mathbf{v}}Z^{\mathbf{u}}$, where
$\mathbf{v},\mathbf{u}\in\{0,1\}^{\otimes n}$ and
$X^{\mathbf{v}}=X_{1}^{v_{1}}X_{2}^{v_{2}}\ldots X_{n}^{v_{n}}$,
$Z^{\mathbf{u}}=Z_{1}^{u_{1}}Z_{2}^{u_{2}}\ldots Z_{n}^{u_{n}}$, can
be mapped, up to a phase, to a quaternary vector, $\mathbf{e}\equiv
\mathbf{u}+\omega \mathbf{v}$, where $\omega^2\equiv
\overline{\omega}\equiv\omega+1$.  A product of two quantum operators
corresponds to a sum ($\bmod\, 2$) of the corresponding vectors.  Two
Pauli operators commute if and only if the \emph{trace inner product}
$\mathbf{e}_1 * \mathbf{e}_2 \equiv \mathbf{e}_1 \cdot
\overline{\mathbf{e}}_2 + \overline{\mathbf{e}}_1 \cdot \mathbf{e}_2$
of the corresponding vectors is zero, where $\overline{\mathbf{e}}
\equiv \mathbf{u} + \overline\omega \mathbf{v}$.

With this map, generators of a stabilizer group are mapped to rows of
a parity check matrix $H$ of an \emph{additive} (forming a group with
respect to addition but not necessarily over the full set of
$\mathbb{F}_4$ operations) code over $\mathbb{F}_4$, with the
condition that the trace inner product of any two rows
vanishes\cite{Calderbank-1997}.  The vectors generated by rows of $H$
correspond to stabilizer generators which act trivially on the code;
these vectors form the \emph{degeneracy group} and are omitted from
the distance calculation.  For a more narrow set of CSS codes the
parity check matrix is a direct sum $H=G_x\oplus \omega G_z$, and the
commutativity condition simplifies to $G_{x}G_{z}^{T}=0$.

An LDPC code, quantum or classical, is a code with a sparce parity
check matrix.  For a \emph{regular} $(j,l)$ LDPC code, every column
and every row of $H$ have weights $j$ and $l$ respectively, while for
a $(j,l)$-limited LDPC code these weigths are limited from above by
$j$ and $l$.  

The huge advantage of classical LDPC codes is that they can be decoded
in linear time using belief propagation (BP) and related iterative
methods\cite{Gallager1962,MacKay:2002}.  Unfortunately, this is not
necessarily the case for quantum LDPC codes: Tanner graphs for quantum
codes have many short loops of length $4$, which cause a dramatic
deterioration of the convergence of the BP
algorithm\cite{Poulin-Chung-2008}.  This problem can be circumvented
with specially designed quantum
codes\cite{Kasai-Hagiwara-Imai-Sakaniwa-2012,Andriyanova-Maurice-Tillich-2012},
but a general solution is not known.  One alternative which has
polynomial complexity is $n$, approaching linear for very small error
rates, is the cluster-based decoding suggested in
Ref.~\cite{Kovalev-Pryadko-FT-2012}.

\section{Generic numerical techniques for distance calculation}
\label{sec:generic}
The problem of numerically calculating the distance of a linear code (finding
the minimum-weight codeword in the code) is related to the decoding problem:
find the most likely (minimum-weight in the case of the $q$-ary symmetric
channel) error which gives the same syndrome as the received codeword.  The
number of required steps $N$ usually scales exponentially with the blocklength
$n$, $N\propto q^{F n}$; we characterize the complexity by the exponent $F$.
For example, for a linear $q$-ary code with $k$ information qubits, there are
$q^k$ distinct codewods, going over each of the codewords has the
complexity exponent 
$F=R$, where $R=k/n$ is the code rate.  When used for decoding, one can
instead store the list of all $q^{n-k}$ syndromes and coset leaders, which corresponds to the
complexity $F=1-R$.

\subsection{Sliding window technique}
\label{sec:sliding} 
This decoding technique has been
proposed in Ref.~\cite{Evseev-1983}, and generalized in
Ref.~\cite{Dumer-1996}.
A related technique has also been independently invented in
Refs.~\cite{Zimmermann-1996,Grassl-2006}. 
For a $q$-ary code with relative distance $\delta\equiv d/n$, the
complexity exponent
is $F_A=R H_q(\delta)$, where $H_q(x)=x\log_q(q-1)-x\log_qx-(1-x)\log_q(1-x)$ is
the $q$-ary entropy function; for a code with the rate $R\equiv k/n$ on the
Gilbert-Varshamov bound, $R=1-H_q(\delta)$, this gives the complexity exponent
$F_A^{(GV)}=R(1-R)$, reaching the maximum of $1/4$ at $R=1/2$.

The idea is to use only $k+o(n)$ consecutive positions to recover any
codeword of a $q$-ary linear $[n,k]$ code. For example, any $k$
consecutive positions suffice in a cyclic code. Similarly, it is easy
to verify that in most (random) $k\times n$ generator matrices $G$
any $s=k+2\left\lfloor \log_{q}n\right\rfloor $ consecutive columns
form a submatrix $G_{s}$ of a maximum rank $k$. Thus, $s$ (error free)
consecutive bits suffice to recover a codeword in most random $[n,k]$
codes.

To find a codeword $c$ of a minimum weight $w$, we choose a sliding
window $I(i,s)$ that begins in a position $i=0,\ldots,n-1$ and has
length $s$.  Our goal is to find the window that has the average
Hamming weight, $v\equiv\left\lfloor
  ws/n\right\rfloor $.  (Note that a sliding window can change its
weight only by one when it moves from any position $i$ to $i+1$; thus
at least one of the $n$ windows will have weight $v$.)  For each
$i$ and for each $w=1,2,\ldots$, we encode all possible
\begin{equation}
L=(q-1)^{v}\textstyle {s\choose v}
  \label{eq:sliding-wind}
\end{equation}
vectors of length $s$ and weight $v$.  We stop the procedure
once we find an encoded codeword of weight $w$. The overall procedure
has complexity of the order $Ln^{2}\asymp q^{F_A n}$, where
$F_A=RH_{q}(\delta)$. 

Unfortunately, the performance suffers when the technique is applied
to a quantum code.  Indeed, the additive quaternary code corresponding
to an $[[n,k]]$ stabilizer code operates in a space with $4^n$ symbols
with only $2^{r}=4^{r/2}$ distinct syndromes, where $r\equiv n-k$ is
the redundancy of the quantum code; the effective rate is
thus\footnote{This construction is analogous to pseudogenerators
  introduced in Ref.~\cite{White-Grassl-2006}.}
$R'=(n-r/2)/n=(1+R)/2$.  The same effective rate is obtained if we
take a CSS code with $\rank G_x=\rank G_z=(n-k)/2$, as there are
$k'=n-(n-k)/2=(n+k)/2$ information bits for both codes.  In addition
to an increased number of the information symbols, each obtained
vector of small weight has to be tested on linear dependence with the
rows of the parity check matrix $H$.  In addition to the considered
mapping, one can also map an additive $[[n,k,d]]$ code to a binary
code with block length $3n$ and weight of each codeword doubled; such 
mapping typically gives a larger complexity and will not be considered
here\cite{Calderbank-1997,White-Grassl-2006}.

For a generic stabilizer code with relative distance $\delta$, the
binary complexity exponent of the sliding-window technique is $F=2 R'
H_4(\delta)$.  Similarly, for a CSS code, the sliding-window technique
gives the complexity exponent $F_{Aq}=2 R' H_2(\delta)$.  Both results
produce the same complexity exponent $$F_{Aq}^{(GV)}=(1-R^2)/2$$ on the
quantum GV bound, namely $R=1-2H_4(\delta)$ for generic quantum
codes\cite{Feng:dec.2004}, and $R=1-2H_2(\delta)$ for CSS
codes\cite{Calderbank-Shor-1996}.  The dependence $F_{Aq}^{(GV)}(R)$
is shown in Fig.~\ref{fig:cmp} with a solid red line.  Note that for
codes with small relative distance $\delta$, the complexity exponent
is logarithmic in $\delta$, e.g., $F_{Aq}\sim \delta
(1+R)\log_{2}(e/\delta)$ in the case of a CSS code.

\begin{figure}[htbp]
  \centering
  \includegraphics[width=\columnwidth]{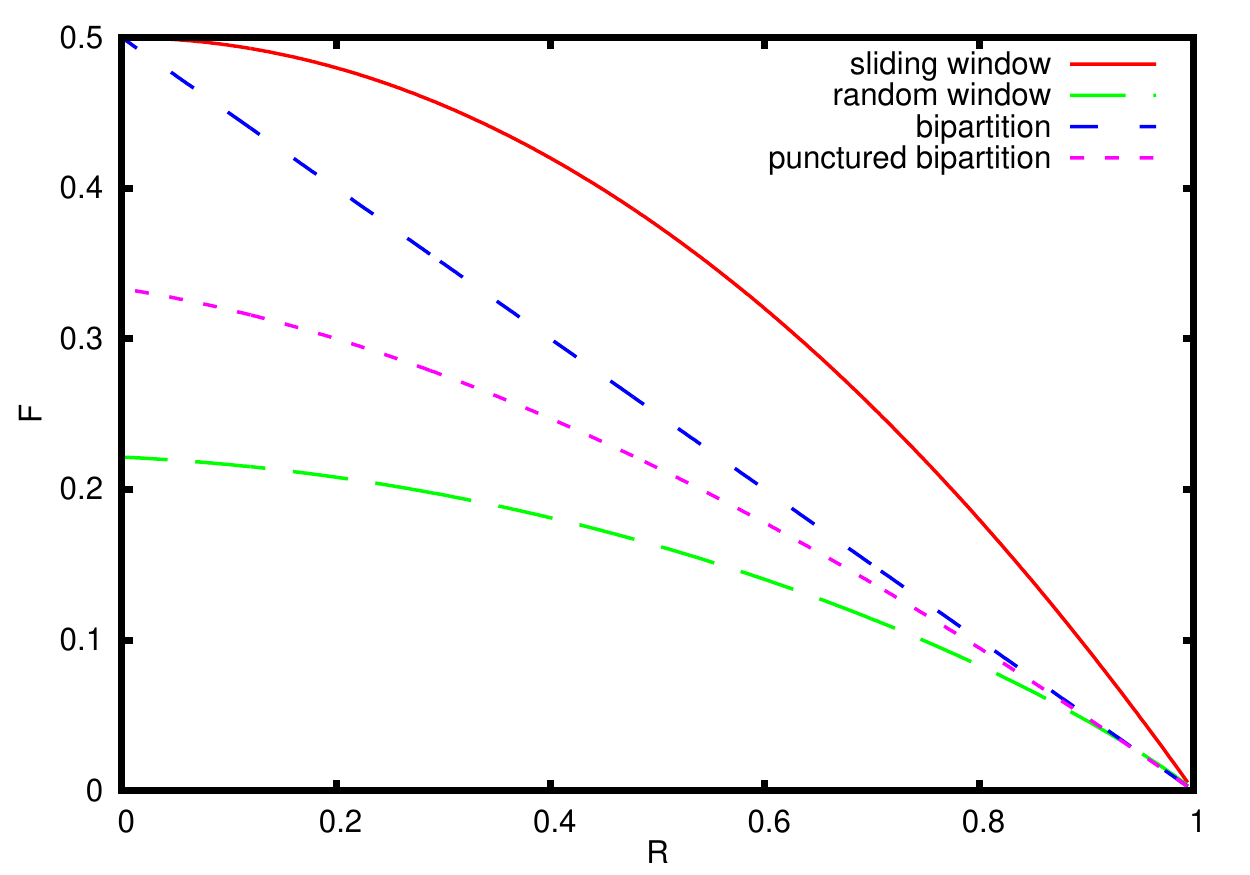}
  \caption{Comparison of the binary complexity exponents for the four
    classical decoding techniques applied to quantum codes at the
    quantum GV bound, see Sec.~\ref{sec:generic}.  Note that for
    high-rate codes, $R\to1$, the curves for the sliding window and
    the random window techniques have logarithmically-divergent slopes,
    while the slopes for the two other techniques remain finite.  In
    this limit of $R\to1$ the punctured bipartition technique gives the best
    performance.}
  \label{fig:cmp}
\end{figure}

\subsection{Random window
  technique\cite{Leon-1988,Stern1989,Kruk-1989,Coffey-Goodman-1990}} 

Given a $q$-ary linear $(n,k)$ code, we randomly choose $s=k+\tau$ positions
with positive number $\tau=o(k)$.  For any codeword, we wish to find an
$s$-set of small weight $t$.  Given an error pattern (or a codeword) of weight
$w$, we only need to estimate the number of random trials $T_{t}(n,s,w)$
needed to find such set with a high probability. This number is well known (up
to a factor of order $n)$ \cite{Erdos-book} and it is upper-bounded by
%%% Given a $q$-ary linear $[n,k]$ code, we randomly choose $s=k$ [more
%%%   generally, $s=k+\tau$ where $\tau=o(k)]$ positions. The idea is to
%%% find an information $s$ set, which is a set of $s$ information
%%% positions that have a very low number of errors $t$, say, $t=1$. For
%%% any such $ s$ set, we can then consider $Q_{1}=s$
%%% possible error-free vectors and re-encode them into the codewords of
%%% length $n$ (it is not necessary to vary the only symbol as it does not
%%% affect the weight of the resulting vectors). Note also that a randomly
%%% chosen $s$-set yields a full-rank random submatrix $G_{s}$ with a high
%%% probability $1-q^{-\tau}$. %%%  (more precisely, most codes have
%%% %%% \textit{any} random matrix $G_{s}$ of rank $k-n^{1/2}$ or more.)
%%% Also, given an error pattern (or a codeword) of weight $w$, we can
%%% tightly estimate the number of random trials $T_{t}(n,s,w,t)$ needed
%%% to find at least one $(s,t)$ set with high probability
%%% $1-\exp\{-n\}$. This number is well known (up to a factor of order
%%% $n\ln n$) \cite{Erdos-book} and achieves its maximum at $t=0$:
\begin{equation}
  T_{0}(n,s,w)\asymp \textstyle{n\choose w} /{n-s\choose w}.\label{eq:112} 
\end{equation}
To determine the distance of a code, we choose $w=1,2,\ldots$.  Then
we perform $nT_{0}(n,s,w)$ trials of choosing $s$ random positions. 

Note that a randomly
chosen $k\times s$ submatrix $G_{s}$ of a random generator matrix $G$ has full
rank $k$ with a high probability $1-q^{-\tau}$ (also, most matrices $G$ have
\textit{all possible } submatrices $G_{k}$ of rank $k-n^{1/2}$ or more$).$ If
the current $s$-set includes any information $k$-subset, we only consider $s$
vectors $(0...010...0)$ of weight  $t=1$. We then  re-encode them into the
codewords of length $n.$ Otherwise, we discard an $s$-set and proceed further.
We stop the algorithm, once we obtain a codeword of weight $w.$ 
The overall
complexity has the order of $n^{4} T_{0}(n,s,w)$, and it is
independent of $q$, in contrast to the sliding-window technique. 
This corresponds to the binary complexity exponent
$F_B=H_2(\delta)-(1-R)H_2(\delta/(1-R))$.  For binary random linear codes
meeting the GV bound, the complexity exponent is easily verified to be
\[
F_B^{(GV)}=(1-R)\left(  1-H_2\left(  \delta/(1-R\right)  \right)  ,
\]
where $H_2(\cdot)$ is the binary entropy and $\delta\equiv H_2^{-1}(1-R)$ is
the relative GV distance $w/n$.  In particular, $F\approx0.11$ for the
rate $R=1/2$.    For small
$w\leq(n-k)^{1/2}$, we can also use a simpler estimate 
\[
T_{0}(n,s,w)\asymp\left(  \frac{n}{n-s}\right)  ^{w}\asymp(1-R)^{-w}
\]
that has the exponent  linear in code distance $w$.

Just as the sliding window technique, this technique relies on
recoding (re-encoding).  Thus, the quantum complexity can be obtained
by substituting the effective rate $R'=(1+R)/2$.  In particular, for a
generic stabilizer code meeting the quantum GV bound, we have the
binary complexity exponent as shown in Fig.~\ref{fig:cmp} with the
green dashed line; it reaches the maximum of $F_{\rm max}\approx 0.22
$ at $R=0$, i.e., for small-rate codes.

%\textit{3) \cite{dum2}}: 
\subsection{Bipartition technique\cite{Dumer-1989}}
\label{sec:bipartition} The idea is to use a sliding (``left'') window
of length $s_{l}=\left\lfloor n/2\right\rfloor $ starting in any
position $i$.  For any vector of weight $w$, at least one position $i$
will produce a window of weight $v_{l}=\left\lfloor w/2\right\rfloor
$. The remaining (right) window of length $s_{r}=\left\lceil
  n/2\right\rceil $ will have the weight $v_{r}=\left\lceil
  w/2\right\rceil $.  We calculate the syndromes of all vectors
$e_{l}$ and $e_{r}$ of weights $v_{l}$ and $v_{r}$ on the left and
right windows, respectively, and try to find a pair of vectors
$\{e_{l},e_{r}\}$ that produce identical syndromes, and therefore form
a codeword. Clearly, each set $\{e_{l}\}$ and $\ \{e_{r}\}$ have
exponential size of order $L=(q-1)^{w/2}{n/2\choose w/2} $.  Finding
two elements $e_{l}$, $e_{r}$ with equal syndromes can be performed,
e.g., by sorting the elements of the combined set, or, to save on
memory, sorting the elements of the left set and using binary search
for each of the syndromes from the right set.  This has a similar
complexity of order $L\log_{2}L$. Thus, finding a code vector of
weight $w=\delta n$ requires complexity of order
\[
q^{F_Cn},\;F_C=H_{q}(\delta)/2.
\]
For random binary codes which meet the GV bound, we have exponent
$F_{C}^{(GV)}=(1-R)/2$. For code rate $R=1/2$, this gives
$F_{C}^{(GV)}= F_{A}^{(GV)}=1/4$. For higher code rates, the
bipartition technique gives exponent $F_{C}^{(GV)}<F_{A}^{(GV)}$.  It
can also be verified that $F_{C}^{(GV)}<F_{B}^{(GV)}$ for code rates
$R$ approaching 1.  In addition, bipartition technique is guaranteed
to work with any linear code, as opposed to two previous techniques
provably valid for random codes.

The bipartition technique is also the only technique that can be
transferred to quantum codes without any performance loss.  For a
generic quantum code and a CSS code corresponding binary complexity
exponents are $F_{Cq}=H_4(\delta)$ and $H_2(\delta)$, respectively.  On the
quantum GV bound, this gives binary exponent $F_{Cq}^{(GV)}=(1-R)/2$ in both
cases, see Fig.~\ref{fig:cmp}.  Note that this line is always below
that for $F_{Aq}^{(GV)}$, and for high-rate codes the
corresponding line is below that for the random window technique, $F_{Bq}^{(GV)}$.

\subsection{Punctured bipartition technique \cite{Dumer-2001}}

Here we combine the sliding-window technique with bipartition.
Consider a relatively large sliding window of length
\begin{equation}
s=\left\lceil
  2nR/(1+R)\right\rceil .\label{eq:punctured-window-size}
\end{equation}
Note that most random $[n,k]$ codes include
at least one information set on any sliding $s$-window $I(i,s)$ with
initial position $i=0,...,n-1$. Thus, any such window forms a
punctured linear $[s,k]$ code with a smaller redundancy $s-k$.  Also,
any codeword of weight $w$ has weight $v=\left\lfloor
  ws/n\right\rfloor $ on some sliding window.  For simplicity, let $s$
and $v$ be even. We then use bipartition on each $s$-window and
consider all vectors $e_{l}$  and $e_{r}$  of weight $v/2$ on either
half of length $s/2$. The corresponding sets $\{e_{l}\}$ and $
\{e_{r}\}$ have size $L_{s}%
=(q-1)^{v/2}\left( _{v/2}^{s/2}\right) $. We then seek all matching
pairs $\{e_{l},e_{r}\}$ that have the same syndrome $h$.  Each such
pair $\{e_{l},e_{r}\}$ represents some code vector of the punctured
$[s,k]$ code and is re-encoded to the full length $n$. For each
$w=1,2,...,$ we stop the procedure once we find a re-encoded vector of
weight $w$. Obviously, this technique can lower the complexity to the
order $L_{s}$. Note, however, that many vectors $e_{l}$ and $e_{r}$ of
length $s/2$ can simultaneously have the same syndrome $h$ of size
$s-k$. Thus, our task is to encode \textit{all code vectors} of weight
$v$ in a random $[s,k]$ code. It can be shown \cite{Dumer-2001} that
our choice of parameter $s$ limits the number of such codewords by the
same order $L_{s}$. Thus, we can find any codeword of weight $w=\delta
n$ with a smaller complexity
\[
q^{F_{C}s}=q^{F_{D}n},\;F_{D}=H_{q}(\delta)R/(1+R).
\]
For codes meeting the GV bound, $F_{D}^{(GV)}=R(1-R)/(1+R)$. Note however,
that this combined technique cannot be provably applied to any 
linear code, in contrast to a simpler bipartition technique.

Somewhat similarly to regular case in Sec.~\ref{sec:bipartition}, the
performance of the bipartition in this technique is not affected when
we consider quantum codes.  However, in the
expression~(\ref{eq:punctured-window-size}) for the optimal block
size, one needs to use the effective quantum rate $R'=(1+R)/2$.  As a
result, the complexity exponent for regular stabilizer codes becomes 
\begin{equation}
  F_{Dq}={2R'\over 1+R'}H_{4}(\delta)={2(1+R)\over
    3+R}H_4(\delta);
  \label{eq:punctured-complexity}  
\end{equation}
 on the GV bound  this gives 
$$
  F_{Dq}^{(GV)}={(1-R^2)\over
    3+R}.
$$
This technique is the best for high-rate quantum codes, $R\to1$.

\section{Linked-cluster technique}

Here we present a technique which is designed specifically for very
sparse quantum LDPC codes, as an alternative to the belief
propagation technique.

For a $(j,\ell)$-limited LDPC code, we represent all (qu)bits as nodes
of a graph ${\cal G}_1$ of degree at most $z$: two nodes are connected
by an edge iff there is a row in the parity check matrix which has
non-zero values at both positions.  An error with support in a subset
${\cal E}\subseteq V({\cal G}_1)$ of the vertices defines the subgraph
${\cal G}_1({\cal E})$ induced by ${\cal E}$.  Generally, we will not
make a distinction between a set of vertices and the corresponding
induced subgraph.  In particular, a (connected) cluster in ${\cal E}$
corresponds to a connected subgraph of ${\cal G}_1({\cal E})$.
Different clusters affect disjoint sets of rows of the parity check
matrix.  This implies the following

\begin{theorem}\label{th:basic}
  The support of a minimum-weight code word of a $q$-ary code with the
  parity check matrix $\mat H$ forms a linked cluster on  ${\cal
    G}_1$.
\end{theorem}
\begin{proof}
  Indeed, let us assume this is not so, and a minimum-weight code word
  $\mathbf{c}$ is supported by two or more disconnected parts.  By
  construction, these affect different rows of the parity check matrix
  and, therefore, the vectors corresponding to subsets of non-zero symbols in
  $\mathbf{c}$ are in the null-space of $\mat H$, contrary to the
  assumption that $\mathbf{c}$ has minimum weight.
\end{proof}
%%% \begin{lemma}
%%%   \label{lemma:cluster-detectable}
%%%   (Lemma 1 from Ref.~\cite{Kovalev-Pryadko-FT-2012})
%%%   For a distance-$d$ LDPC code, any error whose support is a union of
%%%   disconnected clusters on ${\cal G}_1$ of weights $w_i<d$, is
%%%   detectable.
%%% \end{lemma}

Thus, in order to determine the distance of a code by an exhaustive
search, we do not have to list all error patterns; instead, one can go
over all linked clusters of increasing sizes.  We used the following
variant of the breadth-first algorithm to construct all linked cluster
of a given size $w$: 

Start with a position $i=0,1,\ldots,n-w$, and add to the list all
neighboring positions to the right of $i$.  At each subsequent level
of recursion, only go over the positions in the list to the right of
the position added on the previous level.  Once a new position is
selected, add all new neighboring positions which are to the right of
the original starting point $i$.  The recursion should stop after the
desired cluster size $w$ is reached.  This way, the algorithm
generates all clusters of size $w$, and no repeated clusters are
produced.  In the case of a binary code, each linked cluster of weight
$w$ directly corresponds to a potential code word of same weight.  In
the case of a $q$-ary code, one needs to check the rank of a matrix
formed by the corresponding columns of the parity check matrix.

The upper cup on the total number of the linked clusters of size $w$ for a given
$(j,\ell)$-limited LDPC code can be obtained from the cluster
distribution for a regular tree .  The degrees of the graph
${\cal G}_1$ are limited from above by
$z\equiv(\ell-1)j$.   Among the
degree-limited graphs, the $z$-regular tree has the largest number of
clusters (it does not have any loops).  Namely, the number of
weight-$w$ clusters containing a given 
vertex is\cite{Hu-clustersize-1987}
\begin{equation}
  N_w%%%= %w z {((z - 1) m)!\over m!\, ((z - 2) m + 2)!}
  ={z\over  w-1}{(z-1)w\choose w-2}\asymp {z\over (z-2)^2w}2^{(z-1)w
    H_2(1/(z-1))}, 
  \label{eq:cluster-dist}
\end{equation}
where the asymptotic form is valid for large $w$.  Note that the loops
present in the actual graph tend to reduce the exponent; also, at
$w\agt n/(z-1)$ there is further reduction in $N_w$ due to finite-size
effect.  Thus, we expect that for $w\alt n/(z-1)$,  the complexity
exponent for the linked-cluster method can be written as
\begin{equation}
  \label{eq:complexity-exp-cluster}
  F_{\rm LC}=\delta (z_\mathrm{eff}-1)H_2(1/(z_\mathrm{eff}-1))\asymp
  \delta \log_2(e(z_\mathrm{eff}-1)),
\end{equation}
where $z_\mathrm{eff}< z$.  
For example, a number of generalized
hypergraph-product codes have been constructed in
Ref.~\cite{Kovalev-Pryadko-Hyperbicycle-2012} from different binary
cyclic codes.  Codes originating
from the same check polynomial correspond to graphs with  the same
local structure as the graph
${\cal G}_1$ for the original hypergraph-product
codes\cite{Tillich2009}.  Examples of the cluster-number
scaling with weight for several such codes are given in
Fig.~\ref{fig:scaling}.

%%%   While for some such codes the distance is
%%% known exactly, e.g., the original hypergraph-product codes
%%% \cite{Tillich-Zemor-2009} $[[98m^2-42m+9,9,4m]]$, the
%%% hypergraph-product codes from full circulant
%%% matrices\cite{Kovalev-Pryadko-2012} $[[98m^2,18,4m]]$, and the halved
%%% codes \cite{Kovalev-Pryadko-2012} $[[49m^2,18,4m]]$ with even $m$, for
%%% several new families of related codes in
%%% Ref.~\cite{Kovalev-Pryadko-Hyperbicycle-2012} the upper and lower
%%% limits on the distance differ.  In particular, the code $[[294,18,D]]$
%%% has been listed to have $4\le D\le 12$. We used the cluster technique
%%% to show that the distance for this code is actually
%%% $D=8$. (\textbf{???})

\begin{figure}[htbp]
  \centering
  \includegraphics[width=\columnwidth]{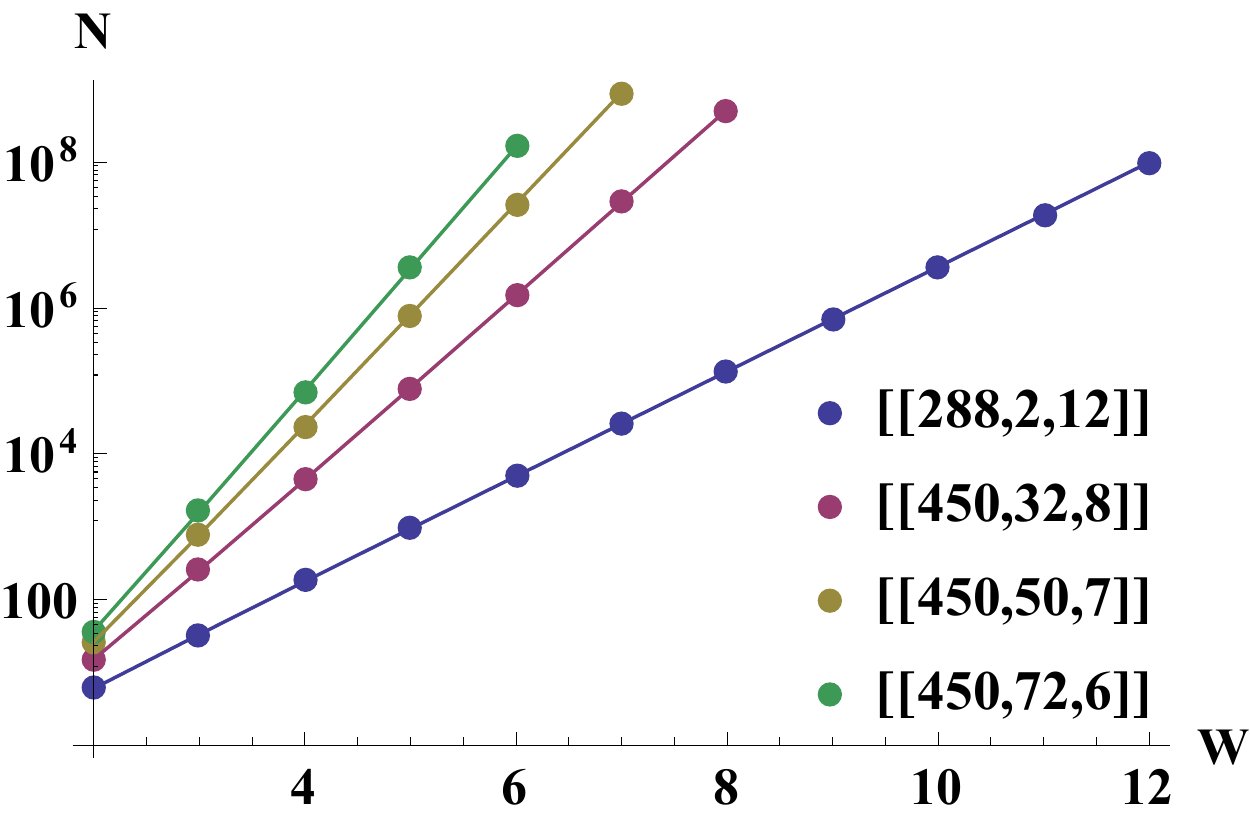}
  \caption{Dependence of the number of clusters $N$ on the
    corresponding weight $w$ for hypergraph-product codes obtained
    from cyclic codes with the check polynomials of weights $w_h= 2,
    3, 4,$ and 5.  The corresponding parity-check matrices are
    $(w_h,2w_h)$-regular, and the graphs ${\cal G}_1$ have degrees
    $z=(l-1)j=(2w_h-1)w_h=6,15,28,$ and 45.  The  fits to 
    $N=Ay^w$ give $y\equiv e(z_\mathrm{eff}-1)=  5.2, 18.8, 33.4$, and 47.0,
    respectively. }
  \label{fig:scaling}
\end{figure}

While the performance of the cluster-based technique deteriorates
rapidly with large $z$, and for larger distances, one advantage
evident from Eq.~(\ref{eq:complexity-exp-cluster}) is that the
complexity exponent is proportional to the relative distance $\delta$.
In comparison, any other deterministic technique in
Sec.~\ref{sec:generic} has the complexity scaling as $F\propto \delta
\log(1/\delta)$ in this limit.  Thus, the presented linked-cluster
technique has the best asymptotic performance for all known quantum
LDPC codes with limited-weight stabilizer generators, where
$\delta\propto n^{-1/2}$.  Compared with the random window technique
(which has the smallest complexity exponent in a wide range of rates),
$F_{Bq}\asymp\delta\log_2(1/(1-R))$ for small relative distances, also
linear in $\delta$, this technique is expected to win at rates such
that $1-R\alt (e \,z_\mathrm{eff})^{-1} $.

%%% Consider 

\section{Conclusion}
We suggested a cluster-based technique for finding the distance of very
sparse quantum LDPC codes.  It beats the existing non-probabilistic
algorithms for codes with sufficiently small relative distances (all
known families of quantum LDPC codes have distance scaling as
$n^{1/2}$ or lower at large $n$).  It also beats the probabilistic
random window technique for codes with sufficiently high rates.


\section*{Acknowledgment}
%We are grateful to I. Dumer and M. Grassl for multiple helpful discussions.
This work was supported in part by the U.S. Army Research Office Grant No.\
W911NF-11-1-0027, and by the NSF Grant No.\ 1018935.

% optional entry into table of contents (if used)
%\addcontentsline{toc}{section}{Acknowledgment}

%%% On the quantum Gilbert-Varshamov boundary, the
%%% com



% http://www.ctan.org/tex-archive/macros/latex/contrib/supported/IEEEtran/bibtex
\bibliographystyle{IEEEtran}

\IEEEtriggeratref{29}
\bibliography{lpp,qc_all,more_qc,MyBIB,ldpc}

\begin{thebibliography}{10}
\providecommand{\url}[1]{#1}
\csname url@samestyle\endcsname
\providecommand{\newblock}{\relax}
\providecommand{\bibinfo}[2]{#2}
\providecommand{\BIBentrySTDinterwordspacing}{\spaceskip=0pt\relax}
\providecommand{\BIBentryALTinterwordstretchfactor}{4}
\providecommand{\BIBentryALTinterwordspacing}{\spaceskip=\fontdimen2\font plus
\BIBentryALTinterwordstretchfactor\fontdimen3\font minus
  \fontdimen4\font\relax}
\providecommand{\BIBforeignlanguage}[2]{{%
\expandafter\ifx\csname l@#1\endcsname\relax
\typeout{** WARNING: IEEEtran.bst: No hyphenation pattern has been}%
\typeout{** loaded for the language `#1'. Using the pattern for}%
\typeout{** the default language instead.}%
\else
\language=\csname l@#1\endcsname
\fi
#2}}
\providecommand{\BIBdecl}{\relax}
\BIBdecl

\bibitem{shor-error-correct}
\BIBentryALTinterwordspacing
P.~W. Shor, ``Scheme for reducing decoherence in quantum computer memory,''
  \emph{Phys. Rev. A}, vol.~52, p. R2493, 1995. [Online]. Available:
  \url{http://link.aps.org/abstract/PRA/v52/pR2493}
\BIBentrySTDinterwordspacing

\bibitem{Knill-Laflamme-1997}
\BIBentryALTinterwordspacing
E.~Knill and R.~Laflamme, ``Theory of quantum error-correcting codes,''
  \emph{Phys. Rev. A}, vol.~55, pp. 900--911, 1997. [Online]. Available:
  \url{http://dx.doi.org/10.1103/PhysRevA.55.900}
\BIBentrySTDinterwordspacing

\bibitem{Bennett-1996}
\BIBentryALTinterwordspacing
C.~Bennett, D.~DiVincenzo, J.~Smolin, and W.~Wootters, ``Mixed state
  entanglement and quantum error correction,'' \emph{Phys. Rev. A}, vol.~54, p.
  3824, 1996. [Online]. Available:
  \url{http://dx.doi.org/10.1103/PhysRevA.54.3824}
\BIBentrySTDinterwordspacing

\bibitem{kitaev-anyons}
\BIBentryALTinterwordspacing
A.~Y. Kitaev, ``Fault-tolerant quantum computation by anyons,'' \emph{Ann.
  Phys.}, vol. 303, p.~2, 2003. [Online]. Available:
  \url{http://arxiv.org/abs/quant-ph/9707021}
\BIBentrySTDinterwordspacing

\bibitem{Dennis-Kitaev-Landahl-Preskill-2002}
\BIBentryALTinterwordspacing
E.~Dennis, A.~Kitaev, A.~Landahl, and J.~Preskill, ``Topological quantum
  memory,'' \emph{J. Math. Phys.}, vol.~43, p. 4452, 2002. [Online]. Available:
  \url{http://dx.doi.org/10.1063/1.1499754}
\BIBentrySTDinterwordspacing

\bibitem{Raussendorf-Harrington-2007}
\BIBentryALTinterwordspacing
R.~Raussendorf and J.~Harrington, ``Fault-tolerant quantum computation with
  high threshold in two dimensions,'' \emph{Phys. Rev. Lett.}, vol.~98, p.
  190504, 2007. [Online]. Available:
  \url{http://link.aps.org/abstract/PRL/v98/e190504}
\BIBentrySTDinterwordspacing

\bibitem{Wang-Fowler-Austin-Hollenberg-Lloyd-2011}
\BIBentryALTinterwordspacing
D.~S. Wang, A.~G. Fowler, and L.~C.~L. Hollenberg, ``Surface code quantum
  computing with error rates over 1\%,'' \emph{Phys. Rev. A}, vol.~83, p.
  020302, Feb 2011. [Online]. Available:
  \url{http://link.aps.org/doi/10.1103/PhysRevA.83.020302}
\BIBentrySTDinterwordspacing

\bibitem{Fowler:PRA2012}
\BIBentryALTinterwordspacing
A.~G. Fowler, M.~Mariantoni, J.~M. Martinis, and A.~N. Cleland, ``Surface
  codes: Towards practical large-scale quantum computation,'' \emph{Phys. Rev.
  A}, vol.~86, p. 032324, Sep 2012. [Online]. Available:
  \url{http://link.aps.org/doi/10.1103/PhysRevA.86.032324}
\BIBentrySTDinterwordspacing

\bibitem{Bombin-PRX-2012}
\BIBentryALTinterwordspacing
H.~Bombin, R.~S. Andrist, M.~Ohzeki, H.~G. Katzgraber, and M.~A.
  Martin-Delgado, ``Strong resilience of topological codes to depolarization,''
  \emph{Phys. Rev. X}, vol.~2, p. 021004, Apr 2012. [Online]. Available:
  \url{http://link.aps.org/doi/10.1103/PhysRevX.2.021004}
\BIBentrySTDinterwordspacing

\bibitem{Bravyi-Poulin-Terhal-2010}
\BIBentryALTinterwordspacing
S.~Bravyi, D.~Poulin, and B.~Terhal, ``Tradeoffs for reliable quantum
  information storage in 2d systems,'' \emph{Phys. Rev. Lett.}, vol. 104, p.
  050503, Feb 2010. [Online]. Available:
  \url{http://link.aps.org/doi/10.1103/PhysRevLett.104.050503}
\BIBentrySTDinterwordspacing

\bibitem{Fowler:2012arXiv}
A.~G. {Fowler}, M.~{Mariantoni}, J.~M. {Martinis}, and A.~N. {Cleland}, ``{A
  primer on surface codes: Developing a machine language for a quantum
  computer},'' \emph{ArXiv e-prints}, Aug. 2012.

\bibitem{Postol-2001}
\BIBentryALTinterwordspacing
M.~S. Postol, ``A proposed quantum low density parity check code,'' 2001,
  unpublished. [Online]. Available: \url{http://arxiv.org/abs/quant-ph/0108131}
\BIBentrySTDinterwordspacing

\bibitem{MacKay-Mitchison-McFadden-2004}
\BIBentryALTinterwordspacing
D.~J.~C. MacKay, G.~Mitchison, and P.~L. McFadden, ``Sparse-graph codes for
  quantum error correction,'' \emph{IEEE Transactions on Information Theory},
  vol.~59, pp. 2315--30, 2004. [Online]. Available:
  \url{http://dx.doi.org/10.1109/TIT.2004.834737}
\BIBentrySTDinterwordspacing

\bibitem{Tillich2009}
J.-P. Tillich and G.~Zemor, ``Quantum ldpc codes with positive rate and minimum
  distance proportional to {$\sqrt{n}$},'' in \emph{Information Theory, 2009.
  ISIT 2009. IEEE International Symposium on}, 28 2009-july 3 2009, pp. 799
  --803.

\bibitem{Kovalev-Pryadko-2012}
A.~A. Kovalev and L.~P. Pryadko, ``Improved quantum hypergraph-product {LDPC}
  codes,'' in \emph{Information Theory Proceedings (ISIT), 2012 IEEE
  International Symposium on}, july 2012, pp. 348--352.

\bibitem{Kovalev-Pryadko-Hyperbicycle-2012}
\BIBentryALTinterwordspacing
------, ``Quantum "hyperbicycle" low-density parity check codes with finite
  rate,'' 2012, unpublished. [Online]. Available:
  \url{http://arxiv.org/abs/1212.6703}
\BIBentrySTDinterwordspacing

\bibitem{Andriyanova-Maurice-Tillich-2012}
I.~Andriyanova, D.~Maurice, and J.-P. Tillich, ``New constructions of {CSS}
  codes obtained by moving to higher alphabets,'' 2012, unpublished.

\bibitem{Kovalev-Pryadko-FT-2012}
\BIBentryALTinterwordspacing
A.~A. Kovalev and L.~P. Pryadko, ``Fault-tolerance of "bad" quantum low-density
  parity check codes,'' 2012, submitted to Phys. Rev. Lett. [Online].
  Available: \url{http://arxiv.org/abs/1208.2317}
\BIBentrySTDinterwordspacing

\bibitem{Loss:PRX2012}
L.~Trifunovic, O.~Dial, M.~Trif, J.~R. Wootton, R.~Abebe, A.~Yacoby, and
  D.~Loss, ``Long-distance spin-spin coupling via floating gates,'' \emph{Phys.
  Rev. X}, vol.~2, p. 011006, Jan 2012.

\bibitem{yamamoto-cnot-2003}
\BIBentryALTinterwordspacing
T.~Yamamoto, Y.~A. Pashkin, O.~Astafiev, Y.~Nakamura, and J.~S. Tsai,
  ``Demonstration of conditional gate operation using superconducting charge
  qubits,'' \emph{Nature}, vol. 425, pp. 941--4, 2003. [Online]. Available:
  \url{http://dx.doi.org/10.1038/nature02015}
\BIBentrySTDinterwordspacing

\bibitem{Martinis-science-2005}
\BIBentryALTinterwordspacing
R.~McDermott, R.~W. Simmonds, M.~Steffen, K.~B. Cooper, K.~Cicak, K.~D. Osborn,
  S.~Oh, D.~P. Pappas, and J.~M. Martinis, ``Simultaneous state measurement of
  coupled josephson phase qubits,'' \emph{Science}, vol. 307, p. 1299, 2005.
  [Online]. Available: \url{http://dx.doi.org/10.1126/science.1107572}
\BIBentrySTDinterwordspacing

\bibitem{Benhelm-2008}
\BIBentryALTinterwordspacing
J.~Benhelm, G.~Kirchmair, C.~F. Roos, and R.~Blatt, ``Towards fault-tolerant
  quantum computing with trapped ions,'' \emph{Nature Physics}, vol.~4, pp.
  463--466, 2008. [Online]. Available: \url{http://dx.doi.org/10.1038/nphys961}
\BIBentrySTDinterwordspacing

\bibitem{Friedenauer-2008}
A.~Friedenauer, H.~Schmitz, J.~T. Glueckert, D.~Porras, and T.~Schaetz,
  ``Simulating a quantum magnet with trapped ions,'' \emph{Nature Physics},
  2008.

\bibitem{Bennett-teleportation-1993}
\BIBentryALTinterwordspacing
C.~H. Bennett, G.~Brassard, C.~Cr\'epeau, R.~Jozsa, A.~Peres, and W.~K.
  Wootters, ``Teleporting an unknown quantum state via dual classical and
  einstein-podolsky-rosen channels,'' \emph{Phys. Rev. Lett.}, vol.~70, pp.
  1895--1899, Mar 1993. [Online]. Available:
  \url{http://link.aps.org/doi/10.1103/PhysRevLett.70.1895}
\BIBentrySTDinterwordspacing

\bibitem{Gottesman-Chuang-1999}
\BIBentryALTinterwordspacing
D.~Gottesman and I.~L. Chuang, ``Demonstrating the viability of universal
  quantum computation using teleportation and single-qubit operations,''
  \emph{Nature}, vol. 402, pp. 390--393, 1999. [Online]. Available:
  \url{http://dx.doi.org/10.1038/46503}
\BIBentrySTDinterwordspacing

\bibitem{De-Pryadko-2012}
\BIBentryALTinterwordspacing
A.~De and L.~P. Pryadko, ``Universal set of scalable dynamically corrected
  gates for quantum error correction with always-on qubit couplings,'' 2013,
  phys. Rev. Letters, to be published. [Online]. Available:
  \url{http://arxiv.org/abs/1209.2764}
\BIBentrySTDinterwordspacing

\bibitem{Gallager1962}
R.~Gallager, ``Low-density parity-check codes,'' \emph{Information Theory, IRE
  Transactions on}, vol.~8, no.~1, pp. 21 --28, january 1962.

\bibitem{MacKay:2002}
D.~J.~C. MacKay, \emph{Information Theory, Inference \& Learning
  Algorithms}.\hskip 1em plus 0.5em minus 0.4em\relax New York, NY, USA:
  Cambridge University Press, 2002.

\bibitem{Poulin-Chung-2008}
D.~Poulin and Y.~Chung, ``On the iterative decoding of sparse quantum codes,''
  \emph{Quant. Info. and Comp.}, vol.~8, p. 987, 2008.

\bibitem{Kasai-Hagiwara-Imai-Sakaniwa-2012}
K.~Kasai, M.~Hagiwara, H.~Imai, and K.~Sakaniwa, ``Quantum error correction
  beyond the bounded distance decoding limit,'' \emph{Information Theory, IEEE
  Transactions on}, vol.~58, no.~2, pp. 1223 --1230, feb. 2012.

\bibitem{Calderbank-1997}
\BIBentryALTinterwordspacing
A.~R. Calderbank, E.~M. Rains, P.~M. Shor, and N.~J.~A. Sloane, ``Quantum error
  correction via codes over {GF(4)},'' \emph{IEEE Trans. Inf. Th.}, vol.~44,
  pp. 1369--1387, 1998. [Online]. Available:
  \url{http://dx.doi.org/10.1109/18.681315}
\BIBentrySTDinterwordspacing

\bibitem{Evseev-1983}
\BIBentryALTinterwordspacing
G.~S. Evseev, ``Complexity of decoding for linear codes.'' \emph{Probl.
  Peredachi Informacii (USSR)}, vol.~19, pp. 3--8, 1983, [Probl. Inf. Transm.
  (USSR), vol. 19, p. 1-6 (1983)]. [Online]. Available:
  \url{http://mi.mathnet.ru/ppi1159}
\BIBentrySTDinterwordspacing

\bibitem{Dumer-1996}
I.~Dumer, ``Suboptimal decoding of linear codes: partition technique,''
  \emph{Information Theory, IEEE Transactions on}, vol.~42, no.~6, pp. 1971
  --1986, nov 1996.

\bibitem{Zimmermann-1996}
K.-H. Zimmermann, ``Integral hecke modules, integral generalized reed-muller
  codes, and linear codes,'' Technische Universit at Hamburg-Harburg, Tech.
  Rep. Tech. Rep. 3-96, 1996.

\bibitem{Grassl-2006}
\BIBentryALTinterwordspacing
M.~Grassl, ``Searching for linear codes with large minimum distance,'' in
  \emph{Discovering Mathematics with Magma}, ser. Algorithms and Computation in
  Mathematics, W.~Bosma and J.~Cannon, Eds.\hskip 1em plus 0.5em minus
  0.4em\relax Springer Berlin Heidelberg, 2006, vol.~19, pp. 287--313.
  [Online]. Available: \url{http://dx.doi.org/10.1007/978-3-540-37634-7_13}
\BIBentrySTDinterwordspacing

\bibitem{White-Grassl-2006}
G.~White and M.~Grassl, ``A new minimum weight algorithm for additive codes,''
  in \emph{Information Theory, 2006 IEEE International Symposium on}, july
  2006, pp. 1119 --1123.

\bibitem{Feng:dec.2004}
K.~Feng and Z.~Ma, ``A finite gilbert-varshamov bound for pure stabilizer
  quantum codes,'' \emph{Information Theory, IEEE Transactions on}, vol.~50,
  no.~12, pp. 3323 -- 3325, dec. 2004.

\bibitem{Calderbank-Shor-1996}
A.~R. Calderbank and P.~W. Shor, ``Good quantum error-correcting codes exist,''
  \emph{Phys. Rev. A}, vol.~54, no.~2, pp. 1098--1105, Aug 1996.

\bibitem{Leon-1988}
J.~S. Leon, ``A probabilistic algorithm for computing minimum weights of large
  error-correcting codes,'' \emph{Information Theory, IEEE Transactions on},
  vol.~34, no.~5, pp. 1354 --1359, sep 1988.

\bibitem{Stern1989}
J.~Stern, ``A method for finding codewords of small weight,'' \emph{in
  \emph{Coding Theory and Applications}, ser. Lecture Notes in Computer
  Science, G.~Cohen and J.~Wolfmann, Eds.}, vol. 388, pp. 106 --113, 1989.

\bibitem{Kruk-1989}
\BIBentryALTinterwordspacing
E.~A. Kruk, ``Decoding complexity bound for linear block codes,'' \emph{Probl.
  Peredachi Inf.}, vol.~25, no.~3, pp. 103--107, 1989, (In Russian). [Online].
  Available: \url{http://mi.mathnet.ru/eng/ppi665}
\BIBentrySTDinterwordspacing

\bibitem{Coffey-Goodman-1990}
J.~T. Coffey and R.~M. Goodman, ``The complexity of information set decoding,''
  \emph{Information Theory, IEEE Transactions on}, vol.~36, no.~5, pp. 1031
  --1037, sep 1990.

\bibitem{Erdos-book}
P.~Erdos and J.~Spencer, \emph{Probabilistic methods in combinatorics}.\hskip
  1em plus 0.5em minus 0.4em\relax Budapest: Akademiai Kiado, 1974.

\bibitem{Dumer-1989}
\BIBentryALTinterwordspacing
I.~I. Dumer, ``Two decoding algorithms for linear codes,'' \emph{Probl.
  Peredachi Inf. (USSR)}, vol.~25, pp. 24--32, 1989, [Probl. Inf. Transm., 25,
  17-23 (1989)]. [Online]. Available: \url{http://mi.mathnet.ru/ppi635}
\BIBentrySTDinterwordspacing

\bibitem{Dumer-2001}
I.~Dumer, ``Soft-decision decoding using punctured codes,'' \emph{IEEE Trans.
  Inform. Theory}, vol.~47, no.~1, pp. 59--71, 2001.

\bibitem{Hu-clustersize-1987}
\BIBentryALTinterwordspacing
C.-K. Hu, ``Exact cluster size distributions and mean cluster sizes for the
  q-state bond-correlated percolation model,'' \emph{Journal of Physics A:
  Mathematical and General}, vol.~20, no.~18, p. 6617, 1987. [Online].
  Available: \url{http://stacks.iop.org/0305-4470/20/i=18/a=059}
\BIBentrySTDinterwordspacing

\end{thebibliography}

\end{document}